\documentclass{article}

\usepackage[nonatbib,preprint]{neurips_2023}
\usepackage{amsmath}
\usepackage{amsthm}
\usepackage{graphicx}
\usepackage{subcaption}
\usepackage{tikz}
\usepackage{booktabs}

\theoremstyle{thmstyleone}%
\newtheorem{theorem}{Theorem}
\newtheorem{lemma}[theorem]{lemma}%

\newtheorem{definition}{Definition}%
\newtheorem{example}{Example}%

\newcommand{\CC}{{\mathbf{C}}}
\newcommand{\RR}{{\mathbf{R}}}
\newcommand{\PP}{{\mathbf{P}}}
\newcommand{\FF}{{\mathbf{F}}}
\newcommand{\ZZ}{{\mathbf Z}}
\newcommand{\CCP}{{\CC\PP}}
\newcommand{\Torus}{{\mathbf{T}}}

\newcommand{\omitt}[1]{}

\usepackage[utf8]{inputenc} 
\usepackage[T1]{fontenc}    
\usepackage{hyperref}       
\usepackage{url}            
\usepackage{booktabs}       
\usepackage{amsfonts}       
\usepackage{nicefrac}       
\usepackage{microtype}      
\usepackage{xcolor}         

\title{Haldane Bundles: A Dataset for Learning to Predict the Chern Number of Line Bundles on the Torus}

\author{%
  Cody Tipton$^1$\thanks{This work was performed during an internship at Pacific Northwest National Laboratory.}, Elizabeth Coda$^2$, Davis Brown$^2$, Alyson Bittner$^2$, Jung Lee$^2$,\\ \textbf{Grayson Jorgenson$^2$, Tegan Emerson$^{2,3,4}$, Henry Kvinge$^{1,2}$} \\
$^1$Department of Mathematics, University of Washington\\
$^2$Pacific Northwest National Laboratory\\
$^3$Department of Mathematics, Colorado State University\\
$^4$Department of Mathematical Sciences, University of Texas\\
\texttt{cat1184@uw.edu}, \texttt{\{first\}.\{last\}@pnnl.gov}
}

\begin{document}

\maketitle

\begin{abstract}
  Characteristic classes, which are abstract topological invariants associated with vector bundles, have become an important notion in modern physics with surprising real-world consequences. As a representative example, the incredible properties of topological insulators, which are insulators in their bulk but conductors on their surface, can be completely characterized by a specific characteristic class associated with their electronic band structure, the first Chern class. Given their importance to next generation computing and the computational challenge of calculating them using first-principles approaches, there is a need to develop machine learning approaches to predict the characteristic classes associated with a material system. To aid in this program we introduce the {\emph{Haldane bundle dataset}}, which consists of synthetically generated complex line bundles on the $2$-torus. We envision this dataset, which is not as challenging as noisy and sparsely measured real-world datasets but (as we show) still difficult for off-the-shelf architectures, to be a testing ground for architectures that incorporate the rich topological and geometric priors underlying characteristic classes. 
\end{abstract}

\section{Introduction}
\label{sect:intro}

As a family of fundamental topological invariants of vector bundles, characteristic classes have long been a central tool in differential topology \cite{milnor1974characteristic}. As is often the case, it was only after they had become an established tool in mathematics that their importance to physics was recognized in, for example, Chern-Simons theory \cite{witten1989quantum,freed1992classical}. Another surprising appearance of characteristic classes is in topological materials, where they can be associated with the electronic band structure of a crystal lattice. Remarkably, non-trivial topology in this setting often has dramatic physical consequences. For example, topological insulators are topologically protected insulators in their interior but conductors on their surface (or edge) \cite{kotetes2019topological}. 

Because materials with topological properties are of intense interest in applications (e.g., next generation computing \cite{legg2022giant} and advanced sensors \cite{liu2014surrounding}), substantial effort has gone into finding them. While this hunt has seen some successes, first-principles approaches to calculating the topological characteristics of a material and then experimentally validating these predictions remain time-, labor-, and resource-intensive. There has thus been interest in leveraging machine learning as a cheap way to predict the topological properties of a material. Most approaches thus far have relied on off-the-shelf architectures and generic data-processing procedures that do not take advantage of the unique data-type that is a vector bundle. However, the problem of predicting the characteristic classes of a material system is well-aligned with current research directions in the field of geometric deep learning (e.g., invariance to Lie group actions). 

We believe that part of the challenge of developing methods of learning from vector bundles arises from the fact that there are no benchmark datasets where the bundle structure is actually exposed. Indeed, the features in large topological materials databases \cite{bradlyn2017topological,marrazzo2019relative} consist solely in the geometric structure and constituent atoms in the crystal lattice with associated atomic characteristics. While topological labels to these datasets ultimately depend on an underlying vector bundle structure, this is several layers beneath the features that are given. To create a dataset where researchers can develop architectures for learning on vector bundles, we introduce the Haldane Bundle dataset, a dataset of complex line bundles on the $2$-torus with labeled first Chern number. Our process of generating valid random line bundles, which is mathematically non-trivial, is inspired by the Haldane model \cite{haldane1988model}, an influential toy model for the anomalous quantum Hall effect. Our contributions in this paper include the following.
\begin{itemize}
\item A description of a method of sampling random line bundles on the $2$-torus and our approach to efficiently calculating the Chern number of these line bundles.
\item Summary statistics of the resulting Haldane Bundles dataset as well as off-the-shelf model performance.
\item A discussion of some of the geometric and topological properties that a model designed for this task should have.
\end{itemize}
We hope that our work will begin the process of building a connection between the geometric deep learning community and scientists studying topological materials. Code to generate our datasets can be found at \url{https://github.com/shadtome/haldane-bundles}.


\section{Related Work}
\label{sect:related-work}

There has been significant interest in using machine learning approaches for materials design and discovery \cite{gubernatis2018machine}. In the realm of topological materials, several works have looked at predicting the topological properties of a material. For example, \cite{claussen2020detection} and \cite{ schleder2021machine} both use tree-based methods that take as input a mix of categorical and numeric features like atomic properties. On the other hand, \cite{zhang2018machine} look at the simple case of predicting topological properties of $1$-dimensional insulators directly from a Hamiltonian. This amounts to calculating the winding number. An off-the-shelf convolutional neural network is used. Finally, \cite{peano2021rapid} try to learn to predict the Hamiltonian which they then diagonalize to get the band structure. None of these papers learn directly from the underlying vector bundles and thus are not able to support development of models that incorporate the mathematical constructions beneath the topological properties of materials. 

While vector bundles are a fundamental concept in both mathematics and physics, they have seen limited use in machine learning. Notable exceptions include \cite{scoccola2023fibered}, where a vector bundle framework is used for dimensionality reduction, \cite{courts2021bundle,coda2022fiber} where a more general fiber bundle framework is used to model many-to-one and many-to-many maps respective, and \cite{kvinge2022neural} where frames in a subbundle of the tangent bundle of data manifolds are used to better understand the learning process in computer vision models.

\section{A Lightening Tour of Vector Bundles, Characteristic Classes, and the Haldane Model}
\label{sect:vector-bundles}


\subsection{Vector bundles}
    
    Vector bundles were introduced to provide a systematic way of studying tangent vectors, differential forms, and other geometric constructions on manifolds. Roughly speaking, a continuous vector bundle is a collection of vector spaces, one attached to each point on a manifold, which vary continuously. The standard example is the tangent bundle of a smooth manifold, which captures the space of tangent vectors as they vary from point to point. More formally, a  {\emph{vector bundle of rank $n$}} on a manifold $M$ is a space $E$ with a surjective map (the bundle {\emph{projection}}) $\pi:E \rightarrow M$ such that for any $x \in M$, there is a neighborhood $U \subseteq M$ and a homeomorphism $\phi: \pi^{-1}(U) \rightarrow U \times \mathbb{R}^n$. $\pi$ further satisfies a commutative diagram that we omit here but that can be informally interpreted as saying that $\pi$ maps $U \times \mathbb{R}^n$ to $U$. That is, $\pi$ collapses all the points in the vector space $\mathbb{R}^n$. One consequence of this definition is that for any $x \in M$, $\pi^{-1}(x) = \mathbb{R}^n$. This space is known as the {\emph{fiber}} at $x$. Vector bundles can be defined for both fields of $\RR$ and fields over $\CC$ and this choice has a substantial impact on the structure and properties of the vector bundle.

    Vector bundles can have non-trivial topology introduced by ``twists'' in the way that the vector spaces (sitting above points in the manifold) change as one moves around the manifold. Line bundles on the circle provide a visualization of this phenomenon (see Figure \ref{fig: circle vector bundles}). In both the left and right example, lines ($1$-dimensional vector spaces) are assigned to each point on the circle $S^1$ (shown in the bottom of the figure). In both cases, if one only looks at the points along a small interval of $(\alpha,\beta) \subset S^1$ (thought as an open subset of angles between $\alpha$ and $\beta$) along with their associated lines, the resulting space is topologically equivalent to $(\alpha,\beta) \times \mathbb{R}$. That is, the vector bundle is trivial locally. But looking at all of $S^1$ one sees that in the example on the right (which is a M\"{o}bius band), there is a twist such that if we start at a non-zero point on a fiber and then travel around the circle and come back to the same fiber, we will find ourself at a different point in that fiber. This shows that it is impossible to define a continuous function from the whole circle to the M\"{o}bius bundle that does not have any zeros.

    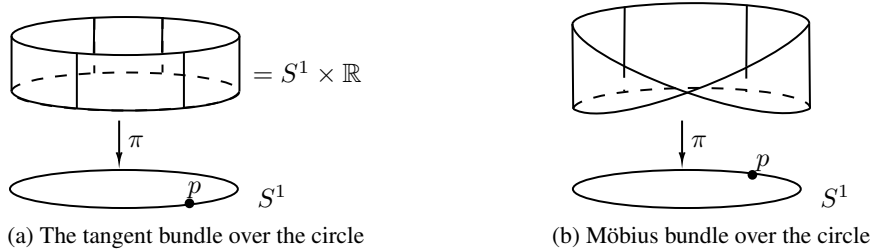
\begin{figure}[h!]
        \centering

\begin{subfigure}{0.5\textwidth}
\centering
\tikzset{every picture/.style={line width=0.75pt}} 

\begin{tikzpicture}[x=0.3pt,y=0.3pt,yscale=-1,xscale=1]

\draw   (152.67,73) .. controls (152.67,59.75) and (216.99,49) .. (296.33,49) .. controls (375.68,49) and (440,59.75) .. (440,73) .. controls (440,86.25) and (375.68,97) .. (296.33,97) .. controls (216.99,97) and (152.67,86.25) .. (152.67,73) -- cycle ;
\draw  [dash pattern={on 4.5pt off 4.5pt}] (152.67,143) .. controls (152.67,129.75) and (216.99,119) .. (296.33,119) .. controls (375.68,119) and (440,129.75) .. (440,143) .. controls (440,156.25) and (375.68,167) .. (296.33,167) .. controls (216.99,167) and (152.67,156.25) .. (152.67,143) -- cycle ;
\draw    (152.67,73) -- (152.67,143) ;
\draw    (440,73) -- (440,143) ;
\draw   (150.67,266) .. controls (150.67,252.75) and (214.99,242) .. (294.33,242) .. controls (373.68,242) and (438,252.75) .. (438,266) .. controls (438,279.25) and (373.68,290) .. (294.33,290) .. controls (214.99,290) and (150.67,279.25) .. (150.67,266) -- cycle ;
\draw    (288.33,179.67) -- (288.33,228.67) ;
\draw [shift={(288.33,230.67)}, rotate = 270] [color={rgb, 255:red, 0; green, 0; blue, 0 }  ][line width=0.75]    (10.93,-3.29) .. controls (6.95,-1.4) and (3.31,-0.3) .. (0,0) .. controls (3.31,0.3) and (6.95,1.4) .. (10.93,3.29)   ;
\draw  [fill={rgb, 255:red, 0; green, 0; blue, 0 }  ,fill opacity=1 ] (372.17,284.17) .. controls (372.17,281.5) and (374.33,279.33) .. (377,279.33) .. controls (379.67,279.33) and (381.83,281.5) .. (381.83,284.17) .. controls (381.83,286.84) and (379.67,289) .. (377,289) .. controls (374.33,289) and (372.17,286.84) .. (372.17,284.17) -- cycle ;
\draw    (370,94) -- (370,164) ;
\draw    (235,94) -- (235,164) ;
\draw  [dash pattern={on 4.5pt off 4.5pt}]  (257,50) -- (257,120) ;
\draw  [dash pattern={on 4.5pt off 4.5pt}]  (343,49) -- (343,119) ;
\draw    (152.67,143) .. controls (165.33,173) and (423.33,177) .. (440,143) ;
\draw    (257,50) -- (256.33,95) ;
\draw    (343,49) -- (342.33,94) ;

\draw (296,195.4) node [anchor=north west][inner sep=0.75pt]    {$\pi $};
\draw (459,257.4) node [anchor=north west][inner sep=0.75pt]    {$S^{1}$};
\draw (453,100.4) node [anchor=north west][inner sep=0.75pt]    {$=S^{1} \times \mathbb{R}$};
\draw (371,252.4) node [anchor=north west][inner sep=0.75pt]    {$p$};

\end{tikzpicture}

        \caption{The tangent bundle over the circle}
        \label{fig:tangent}
\end{subfigure}%
\begin{subfigure}{0.5\textwidth}
\centering

\tikzset{every picture/.style={line width=0.75pt}} 

\begin{tikzpicture}[x=0.3pt,y=0.3pt,yscale=-1,xscale=1]

\draw  [draw opacity=0] (140.07,52.8) .. controls (142.42,38.7) and (208.24,27) .. (289.17,26.49) .. controls (371.64,25.97) and (438.57,37.27) .. (438.66,51.73) .. controls (438.66,51.99) and (438.64,52.25) .. (438.6,52.51) -- (289.33,52.67) -- cycle ; \draw   (140.07,52.8) .. controls (142.42,38.7) and (208.24,27) .. (289.17,26.49) .. controls (371.64,25.97) and (438.57,37.27) .. (438.66,51.73) .. controls (438.66,51.99) and (438.64,52.25) .. (438.6,52.51) ;  
\draw  [draw opacity=0][dash pattern={on 4.5pt off 4.5pt}] (141.07,159.8) .. controls (143.42,145.7) and (209.24,134) .. (290.17,133.49) .. controls (372.64,132.97) and (439.57,144.27) .. (439.66,158.73) .. controls (439.66,158.99) and (439.64,159.25) .. (439.6,159.51) -- (290.33,159.67) -- cycle ; \draw  [dash pattern={on 4.5pt off 4.5pt}] (141.07,159.8) .. controls (143.42,145.7) and (209.24,134) .. (290.17,133.49) .. controls (372.64,132.97) and (439.57,144.27) .. (439.66,158.73) .. controls (439.66,158.99) and (439.64,159.25) .. (439.6,159.51) ;  
\draw    (141.07,159.8) .. controls (189.33,189.67) and (409.33,99.67) .. (438.6,52.51) ;
\draw    (140.07,52.8) .. controls (141.33,94.67) and (379.33,189.67) .. (439.6,159.51) ;
\draw    (140.07,52.8) -- (141.07,159.8) ;
\draw    (438.6,52.51) -- (439.6,159.51) ;
\draw    (206,30.67) -- (205.33,105.67) ;
\draw  [dash pattern={on 4.5pt off 4.5pt}]  (205.33,105.67) -- (205.33,137.67) ;
\draw    (360,28.67) -- (359.33,106.67) ;
\draw  [dash pattern={on 4.5pt off 4.5pt}]  (359.33,106.67) -- (359.33,138.67) ;
\draw   (140.67,261) .. controls (140.67,247.75) and (204.99,237) .. (284.33,237) .. controls (363.68,237) and (428,247.75) .. (428,261) .. controls (428,274.25) and (363.68,285) .. (284.33,285) .. controls (204.99,285) and (140.67,274.25) .. (140.67,261) -- cycle ;
\draw    (278.33,174.67) -- (278.33,223.67) ;
\draw [shift={(278.33,225.67)}, rotate = 270] [color={rgb, 255:red, 0; green, 0; blue, 0 }  ][line width=0.75]    (10.93,-3.29) .. controls (6.95,-1.4) and (3.31,-0.3) .. (0,0) .. controls (3.31,0.3) and (6.95,1.4) .. (10.93,3.29)   ;
\draw  [fill={rgb, 255:red, 0; green, 0; blue, 0 }  ,fill opacity=1 ] (362.17,243.17) .. controls (362.17,240.5) and (364.33,238.33) .. (367,238.33) .. controls (369.67,238.33) and (371.83,240.5) .. (371.83,243.17) .. controls (371.83,245.84) and (369.67,248) .. (367,248) .. controls (364.33,248) and (362.17,245.84) .. (362.17,243.17) -- cycle ;

\draw (286,190.4) node [anchor=north west][inner sep=0.75pt]    {$\pi $};
\draw (449,252.4) node [anchor=north west][inner sep=0.75pt]    {$S^{1}$};
\draw (368,213.4) node [anchor=north west][inner sep=0.75pt]    {$p$};

\end{tikzpicture}

        \caption{M\"{o}bius bundle over the circle}
        \label{fig:mobius}
\end{subfigure}
    \caption{There are only two non-isomorphic line bundles over the circle. The left is trivial, the right is not.}
    \label{fig: circle vector bundles}
    \end{figure}

    \subsection{Characteristic classes}
    The example of the cylinder and the M\"{o}bius band in Figure \ref{fig: circle vector bundles} suggests that the extent to which vector bundles twist may be a useful global statistic that can be used to describe a bundle. This and similar notions of vector bundle shape are formalized in the concept of a characteristic class. Characteristic classes are topological invariants that capture global statistics of a vector bundle. They appear in a number of guises and are a fundamental tool in algebraic topology, algebraic geometry, differential geometry, and mathematical physics \cite{milnor1974characteristic}. For example, the Stiefel-Whitney class is a characteristic class of real vector bundles that takes values in $\ZZ/2\ZZ$ and detects whether a vector bundle is orientable. Hence, the cylinder (Figure \ref{fig: circle vector bundles}, left) has trivial Stiefel-Whitney class since it is orientable and the M\"{o}bius band has Stiefel-Whiteny class $\overline{1}$ since it is not orientable. The complex vector bundles in the Haldane Bundles dataset are labeled by a different characteristic called the first Chern class $c_1^\#$ which takes values in $\ZZ$ and detects how far a bundle is twisted from being trivial. 
    
    Even though they are a topological and not a geometric construction, the Chern class of a vector bundle can be calculated via an integral (over the manifold) of a specific differential form $c_1$ called the curvature form. In particular, if $L$ is a line bundle, then 
    \begin{align*}
        c_1^{\#}(L) = \int_M c_1(L).
    \end{align*}
    One method of calculating the Chern number requires one to find a connection of the line bundle and compute its curvature.  For more information about this see, \cite{Chern}.
    
    Finally, it is important to note that a vector bundle is always trivial if it can be defined consistently with a global set of coordinates. This explains why in Section \ref{sect: construction}, we define line bundles in local patches which we then glue together. If we did not, every vector bundle we defined would be trivial (equivalent to $M \times \RR$).
    
    \subsection{The Haldane Model}
    \label{sect:haldane-model}
    
    Our Haldane Bundles dataset is inspired by the Haldane model \cite{haldane1988model}, which is a $2$-dimensional toy model of a Chern insulator (a type of topological insulator whose properties arise from its electronic band structure having non-trivial Chern number). This simplified quantum system is defined over a honeycomb lattice (reminiscent of graphene). Due to the periodic nature of the lattice and Bloch's Theorem, the momentum space for this system can be viewed as sitting on the $2$-torus, $\Torus^2$. Thus, the Hamiltonian is a function of points on $\Torus^2$, $H: \Torus^2 \rightarrow \text{Mat}(\mathbf{C})_{2,2}$. In particular, for $p \in \Torus^2$,
    \begin{equation}\label{eq:haldane_hamiltonian}
        H(p) = \begin{bmatrix}
        G(p) & \overline{F}(p)\\
        F(p) & -G(p)
        \end{bmatrix}
    \end{equation}
    where $F(p) =t_1\sum_{j=1}^3e^{a_i\cdot p}$, $G(p) = M + 2t_2\sum_{j=1}^3\sin(b_i\cdot p)$, 
     $a_i$ are the vectors forming the honeycomb lattice, and the $b_i$ are the second neighbor hopping vectors.
    One can build a $\CC$-line bundle on $\Torus^2$ associated to $H$ by looking at one of the eigenvectors of $H$. Note that besides $p$, several other parameters determine $H$ including $M$, $t_1$, and $t_2$. In the case where $|t_2/M|<\frac{1}{3\sqrt{3}}$, the Chern number is zero.  On the other hand, when $|t_2/M|>\frac{1}{3\sqrt{3}}$, we obtain non-zero Chern numbers. In the latter case we get a model of a topological insulator.
    
    The Haldane model is an attractive starting point for development of methods of learning to predict Chern numbers from line bundles because different line bundles can be explicitly generated by varying $M$, $t_1$, and $t_2$ and based on these choices we know immediately what the Chern number is. Unfortunately, the resulting dataset is too simple for development of generalizable models. Indeed, because of the regularity in the resulting eigenvectors, a model can easily memorize predictive features without any generalization ability. This motivates our introduction of the Haldane Bundles dataset.

\section{The Haldane Bundle Dataset}
\label{sect:haldane-bundles}

In this section, we describe how we construct the class of line bundles on the $2$-torus that are the datapoints of the Haldane Bundles dataset. These line bundles have a range of nice computational and theoretical properties, are easily parameterized since they are fully determined by a pair of smooth functions $G$ and $F$ on $\Torus^2$, yet provide sufficient variation to create a challenging machine learning task. 

    
     
\subsection{Constructing Lots of Line Bundles}\label{sect: construction}
We begin by describing a generic and abstract algorithm for constructing line bundles on any smooth manifold $M$ and then specify to $\Torus^2$. Pick a collection of smooth maps $\psi_{\alpha}:V_{\alpha}\rightarrow \CC^{n+1}\setminus \{0\}$ for $\alpha\in A$ where $A$ is some index set such that $\{V_{\alpha}\}_{\alpha\in A}$ covers the space $M$. The $\{\psi_\alpha\}_{\alpha \in A}$ induce smooth maps $\{\Psi_\alpha\}_{\alpha \in A}$ from each $V_\alpha$ to complex projective space $\CCP^n$, a smooth $2n$-dimensional manifold where each point corresponds to a line in $\CC^{n+1}$. This connection is defined by $\Psi_{\alpha}(p) = [\psi_{\alpha}(p)]$ where $[\psi_{\alpha}(p)]$ means the unique line traveling through the point $\psi_{\alpha}(p)$ and the origin.

The functions $\{\Psi_\alpha\}_{\alpha \in A}$ each independently assign lines to the points in their respective domains $\{V_\alpha\}_{\alpha \in A}$ of $M$. To define a single consistent assignment on $M$, we need to be able to glue the $\{\Psi_\alpha\}_{\alpha \in A}$ together. When $\{\Psi_\alpha\}_{\alpha \in A}$ are sufficiently consistent, we can do this by introducing some gluing maps $\tau_{\alpha,\beta}: V_{\alpha}\cap V_{\beta}\rightarrow \CC^*$ such that $\tau_{\alpha,\beta}(p) \psi_{\beta}(p) = \psi_{\alpha}(p)$ for all $p\in V_{\alpha}\cap V_{\beta}$ and for all $\alpha,\beta\in A$.
These gluing maps ensure that we can consistently translate between $\Psi_\alpha$ to get a well-defined line bundle where pairs of $\Psi_\alpha$ are both defined. Taken together, this construction gives a line bundle that is well-defined on all of $M$. Note that the existence of $\tau_{\alpha,\beta}$ is impossible if either $\psi_{\beta}(p)$ or $\psi_{\beta}(p)$ is zero and the other is not (this will motivate the construction in the next paragraph).

We construct our Haldane Bundles following this method. Let $G:M\rightarrow \RR$, $F:M\rightarrow \CC$ be smooth maps with no common zeros, and define $R,R^{\dagger}:M\rightarrow \RR$ by $R(p) = G(p) + \sqrt{G(p)^2 + F(p)\overline{F(p)}}$ and $R^{\dagger}(p) = G(p) - \sqrt{G(p)^2 + F(p)\overline{F(p)}}$. Let $\psi:V\rightarrow \CC^2\setminus\{0\}$, $\psi^{\dagger}:V^{\dagger}\rightarrow \CC^2\setminus\{0\}$ be defined by 
\begin{align}\label{eq:fourier-poly}
\small
&\psi(p) = \begin{bmatrix}
R(p) \\ F(p)
\end{bmatrix} &\psi^{\dagger}(p) = \begin{bmatrix}-\overline{F(p)} \\ R^{\dagger}(p)  \end{bmatrix}
\end{align}
where $V$ and $V^{\dagger}$ are the open subsets of $M$ where $\psi$ and $\psi^{\dagger}$ are not zero, respectively. Note that $V\cup V^{\dagger}=M$, since $F$ and $G$ have no common zeros. The smooth function $\tau:V\cap V^{\dagger}\rightarrow \CC^*$ defined as $\tau = \frac{R^{\dagger}}{F}$ is a gluing function for the pair of maps $\psi$ and $\psi^{\dagger}$ with $\tau(p)\psi(p) = \psi^{\dagger}(p)$ making the induced map $\Psi:M\rightarrow \CCP^1$, where $\Psi|_{V}=[\psi]$ and $\Psi|_{V^{\dagger}}=[\psi^{\dagger}]$, a well-defined line bundle on $M$ (note that $F$ does not vanish on $V\cap V^{\dagger}$ by construction so $\tau$ is well-defined). 

\omitt{\indent The next lemma gives us a criteria we need to make sure $V$ and $V^{\dagger}$ cover $M$ and hence ensures we obtain a line bundle on $M$.
    \begin{lemma}\label{intersection}
        Given a pair of smooth functions $G:M\rightarrow \RR$ and $F:M\rightarrow \CC$, if $Z(F)\cap Z(G)=\emptyset$, then $Z(\psi)$ and $Z(\psi)$ are disjoint and the collection $\{V,V^{\dagger}\}$ covers $M$.  Furthermore, we have the following descriptions for the zero sets of $\psi$ and $\psi^{\dagger}$:
        \begin{align*}
            &Z(\psi) = \{p\in Z(F) ~:~ G(p) <0\},
            &Z(\psi^{\dagger}) = \{p\in Z(F) ~:~ G(p)>0\}
        \end{align*}
        
    \end{lemma}
    }
We call the pair of functions $G$ and $F$ a \emph{Haldane pair}, and we will call the corresponding line bundle $L(G,F)$, the \emph{Haldane bundle} with respect to $G$ and $F$.


\begin{theorem}\label{thm: Chern number}
    For any Haldane pair $(G,F)$ on a $2$-dimensional smooth manifold $M$, the Chern number for $L(G,F)$ is computed through the following integral
        \begin{align*}
        \small
         c_1^{\#}(L(G,F)) =\frac{i}{2\pi}\left(\int_{V} d(\delta \overline{F})\wedge d(\delta F) + \int_{V^{\dagger}\setminus (\overline{V\cap V^{\dagger}})} d(\delta^{\dagger}F)\wedge d(\delta^{\dagger}\overline{F})\right)
        \end{align*}
        where $\delta = \frac{1}{\sqrt{2R(R-G)}}$ and $\delta^{\dagger} = \frac{1}{\sqrt{2R^{\dagger}(R^{\dagger}-G)}}$.
    \end{theorem}
    
    We provide a proof of Theorem \ref{thm: Chern number} in Section \ref{appendix:proof_thm1} of the Appendix.  

   
    
    \subsection{Over the Torus}\label{sect:torus}
We can realize $\Torus^2$ as the product space $S^1 \times S^1$ and hence to find a function defined on $\Torus^2$ it suffices to find a function of two variables that is doubly periodic. We therefore use Fourier polynomials to construct our Haldane pair. We call a smooth function $f:\Torus^2\rightarrow\CC$ a \emph{complex-valued Fourier polynomial} if for all $p\in\Torus^2$ it takes the form
        \begin{align*}
            f(p) = \frac{1}{2\pi}\sum_{k\in\ZZ^2} c_k e^{2\pi i (k\cdot p)}
        \end{align*}
        where all but a finite number of the $c_k\in\CC$ are zero and $(k\cdot p)$ is the ordinary dot product. Similarly, we call a function $g:\Torus^2\rightarrow \RR$ a real-valued Fourier polynomial if for all $p\in \Torus^2$ we have
        \begin{align*}
            g(p) = \frac{1}{2\pi} \sum_{k\in\ZZ^2} (a_k\cos(2\pi (k\cdot p)) + b_k \sin(2\pi (k\cdot p)))
        \end{align*}
        where all but a finite number of the $a_k\in\RR$ and $b_k\in\RR$ are zero.
        
        Because the torus is equivalent to the product $S^1 \times S^1$, we can represent real-valued functions $f$ on the torus as $n \times n$ $2$-dimensional arrays (generalized images) where $f(i,j) = f(i + n, j + n)$. Note that the condition that makes this different from standard images is that the array "wraps around" on both axes, so that the $(n+1)$th column (respectively row) is equal to the $1$st column (resp. row). This imposes a constraint on the boundary of "images" representing features on the torus.
        
        We can approximate the Chern number on the torus using the formula \ref{thm: Chern number} by partitioning the torus into small enough squares and approximating the integrand on these small patches.  This is possible, since in local coordinates the integrand in the integral is just a sum and product of values of $F,\delta,\delta^{\dagger}$ and partial derivatives of these.  Since our functions are Fourier polynomials, we can compute these easily on a GPU by using discrete convolutions with the coefficients.  For more information on this see \ref{appendix:computation} in the Appendix.
        
        The distribution of Chern numbers that one gets depends on the degree of Fourier polynomials that are used. These distributions are visualized in Figure \ref{fig:chern-distribution} for some small degrees. The time it took us to compute Chern number as a function of Fourier polynomial degree and resolution (of the grid on the torus) on a single GPU (Nvidia Tesla p100 @ 16GB) is shown in Figure \ref{fig:my_label}.

\begin{figure}[h!]
    \centering
    \includegraphics[scale=0.5]{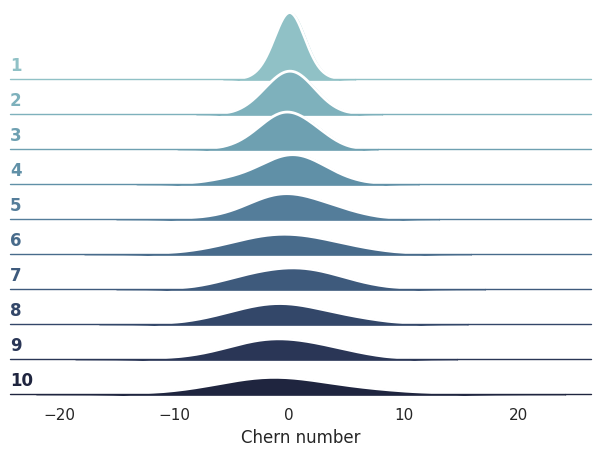}
    \caption{The distribution of Chern numbers for different degrees of Fourier polynomials. The vertical axis is the maximum degree of the polynomials.}
    \label{fig:chern-distribution}
\end{figure}
\begin{figure}
    \centering
    \includegraphics[scale=0.5]{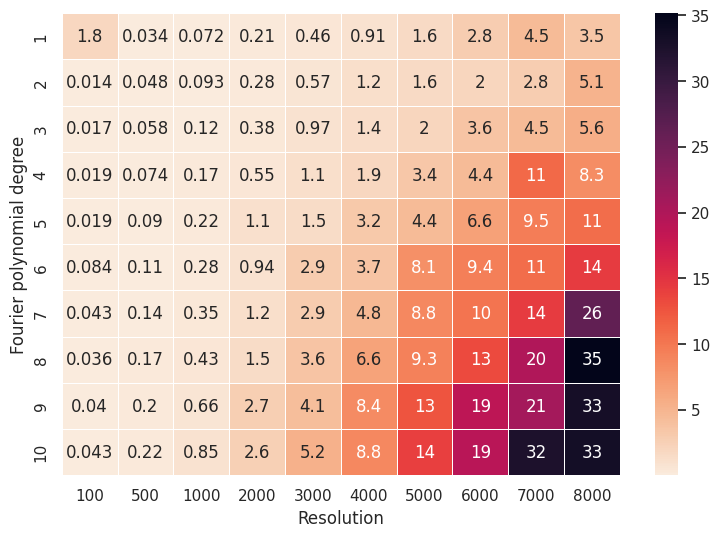}
    \caption{Time to compute the Chern number depending on the max degree of the polynomials and the partition size of the torus, in seconds.}
    \label{fig:my_label}
\end{figure}

\subsection{Is the Haldane Model Enough for an Interesting Machine Learning Task?}
\label{sect:haldane-model-ML}

\begin{figure}[h!]
    \centering
    \includegraphics[scale=0.35]{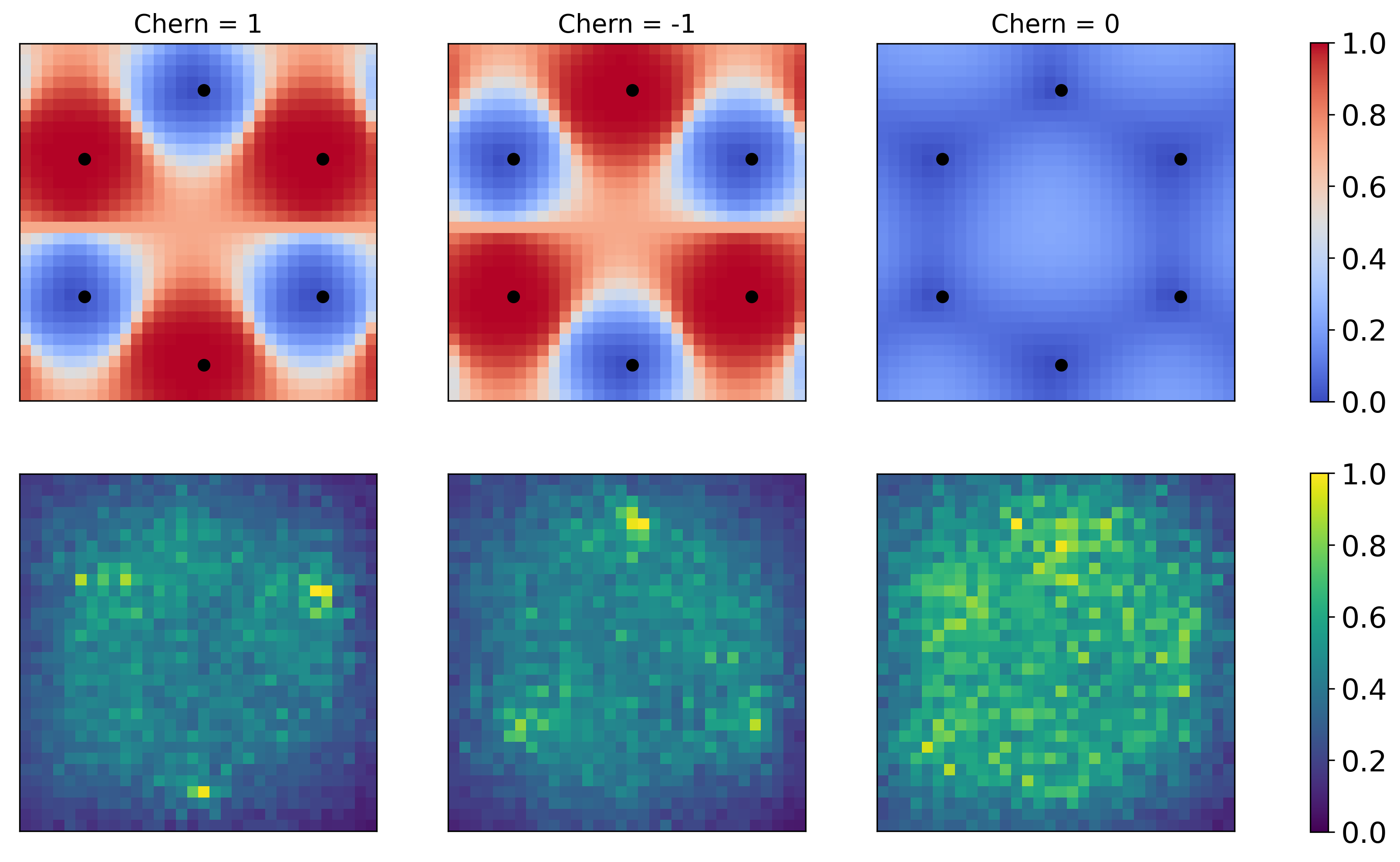}
   \caption{\textbf{Top:} A representative example of one channel of the eigenvector line bundle for the Haldane model for each of the three possible Chern numbers. The plots depict the real part of the first component of the eigenvector. The highlighted points are the zeros of $F$. \textbf{Bottom:} The saliency maps of a trained model, averaged over 100 examples of each class. Observe that for the nontrivial Chern classes, the highlighted regions are centered around the Dirac points.}
   \label{fig:haldane-model}
\end{figure}

   
Before describing the Haldane Bundles dataset described in Section \ref{sect:haldane-bundles}, it is worth asking if all this is necessary. Indeed, might we generate a dataset sufficient for developing a machine learning capability to predict characteristic classes by simulating the simpler Haldane model (described in Section~\ref{sect:haldane-model})? To test this we first fixed the form of $G$ and $F$ in the Hamiltonian in Equation~\ref{eq:haldane_hamiltonian} and varied the parameters $M, t_1,$ and $t_2$ to obtain $10$k distinct Hamiltonians. For each Hamiltonian, we sampled the corresponding eigenvector line bundles uniformly over $\Torus^2$ at a $32 \times 32$ resolution and compute the Chern number (based on the choice of $M$, $t_1$, and $t_2$). Because the line bundles are complex-valued, but high-performing off-the-shelf architectures are real-valued, we convert the complex eigenvector line bundles into separate real and imaginary channels. 

On five random $80/20$ train/test splits for this dataset, a ResNet9~\cite{he2016} achieves $99.8 \pm 0.2$\% test accuracy, averaged across all runs. As derived in ~\cite{Sticlet_2012} from the Brouwer degree formula, the Chern number for the Haldane model can be determined by the sign of $G(p)$ in Equation~\ref{eq:haldane_hamiltonian} at the zeros of $F$. The zeros, or Dirac points, are highlighted in the top row of Figure~\ref{fig:haldane-model}, demonstrating a clear pattern that distinguishes between the different Chern numbers. A simple VanillaGrad saliency map~\cite{Simonyan14a} on the trained model shows that the model exploits this pattern to make its classification. That these points are not present/predictive for arbitrary line bundles yet off-the-shelf models perform with high accuracy suggests that this dataset has too many easily learnable incidental correlations that do not really get at the heart of what Chern numbers are.

\subsection{Generating and Evaluating Machine Learning Models on the Haldane Bundle Dataset}
\begin{figure}[h!]
    \centering
    \includegraphics[scale=0.3]{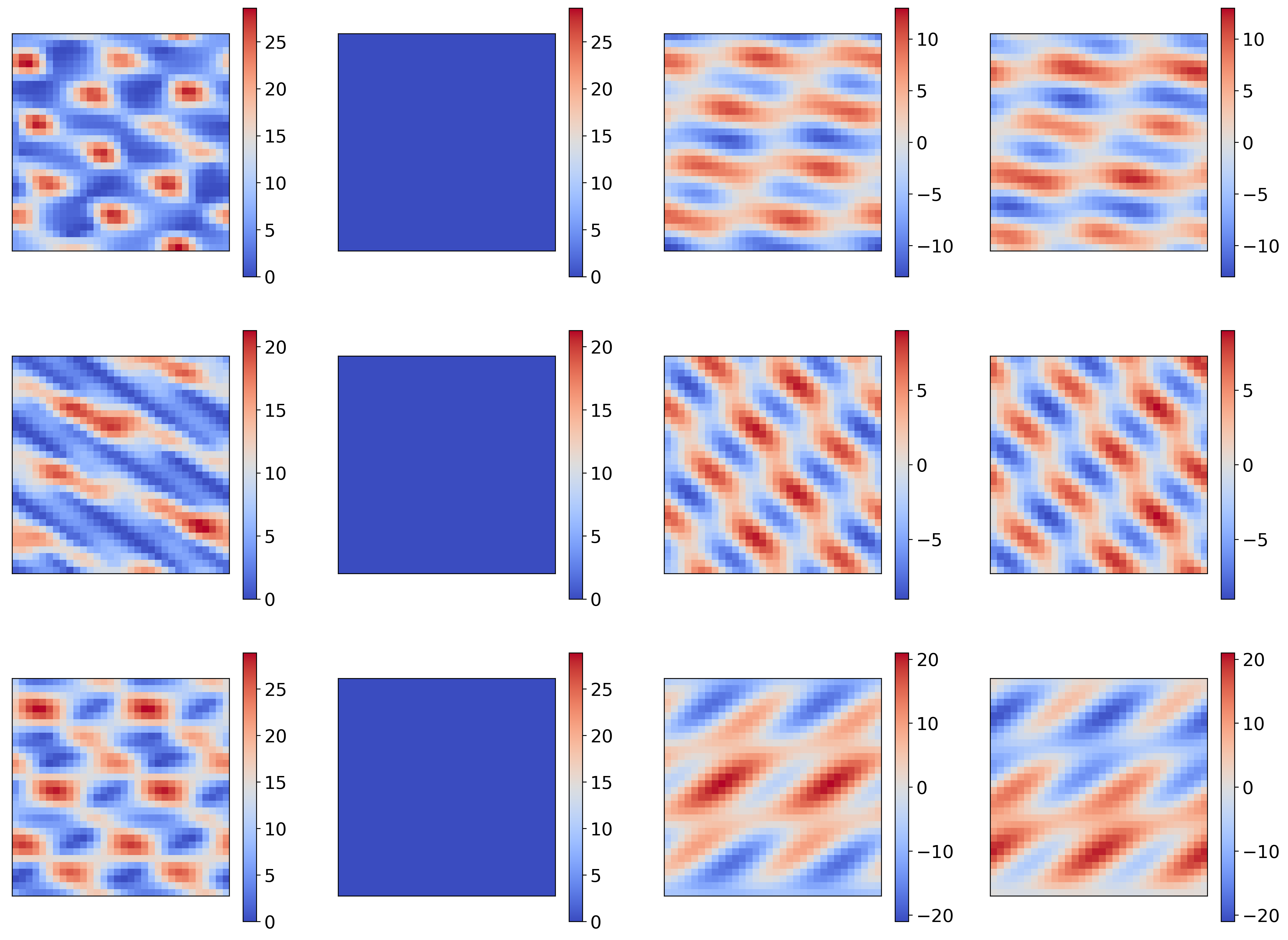}
   \caption{Three examples from the Haldane Bundle dataset with Chern numbers $1,0,$ and $-2$. From left to right the plots depict  $\text{Re}(R(p)), \text{Im}(R(p)), \text{Re}(F(p)),$ and  $\text{Im}(F(p))$ sampled over $\Torus^2$ (where $\text{Re}$ and $\text{Im}$ are the real and imaginary part of a complex number.  }
   \label{fig:random-polynomial}
\end{figure}


To create a more challenging machine learning task, we use the framework from Section \ref{sect: construction} to generate $50$k Fourier polynomial Haldane pairs $(G,F)$ by randomly sampling $a_k,b_k\in\RR$ and $c_k\in\CC$ for $k\in\ZZ^2$ such that $|k_1|\leq 4$ and $|k_2|\leq 4$. We also calculated the Chern number of each Haldane pair. We downsample the generated dataset to obtain a class balanced dataset with $2.5$k examples of each Chern number $ \in \{-2,-1,0,1,2\}$. In contrast to the dataset in the previous section where we fixed the form of $G$ and $F$ and varied parameters to obtain the dataset, here we generate a different form of $G$ and $F$ for each example. As a result, the zeros of $F$ are not spatially fixed within the representation and there is not a trivial pattern of zeros for the model to exploit. Some examples from this dataset are visualized in Figure~\ref{fig:random-polynomial}. Details for the actual calculation can be found in Section \ref{appendix:computation} of the Appendix.

We train a Vision Transformer (ViT)~\cite{Kolesnikov2021}, a ResNet9, and a ResNet9 with standard convolutions replaced by circular convolutions that wrap around to account for the fact that the line bundle is over the torus (see Section \ref{sect:torus} for a discussion of this aspect of the data representation). We report model performance averaged over five random $80/20$ train/test splits in Table~\ref{table:benchmarks}. The best test accuracy any of these standard off-the-shelf models achieved is $31.07 \%$, illustrating the more challenging nature of this dataset and providing evidence that this is a task where some research into novel architectures is warranted.

\begin{table}[h!]
\caption{Benchmarking the Haldane Bundle Dataset on various off-the-shelf models. All accuracies are averaged over five random train/test splits and model initializations. Models are trained on normalized line bundles. $95\%$ confidence intervals are included for each.}
\centering
\begin{tabular}{lrr}
\toprule
    Architecture & Train Accuracy (\%) & Test Accuracy (\%)
 \\ \hline
    ViT & $100.0  \pm  0.0$ & $29.1  \pm  1.16$ \\ 
    ResNet & $99.99  \pm  0.02$ & $26.57  \pm  1.26$  \\ 
    ResNet-C & $100.0  \pm  0.0$ & $26.81  \pm  1.03$ \\ 
    Null Classifier & $20.66  \pm  0.07$ & $21.36  \pm  0.19$  \\ 
\end{tabular}
\label{table:benchmarks}
\end{table}

\section{Topological and Geometric Priors Intrinsic to This Dataset}

We want to end by highlighting some of the interesting geometric and topological features intrinsic to the challenge of trying to predict characteristic classes that may be attractive to the geometric deep learning community and other researchers that are interested in incorporating non-trivial mathematical ideas into deep learning frameworks.

\textbf{Complex-valued deep learning:} The applications that motivated this research (topological materials) and other physics applications where characteristic classes are important all work over the complex numbers. While there is work in complex-valued deep learning, the vast majority of research (including almost all of the science of deep learning) is over $\RR$. On the other hand, Chern class prediction is fundamentally connected to the complex numbers. While our baselines incorporated complex numbers in a naive way, it would be interesting to understand whether deeper integration of complex numbers into the networks would improve performance.

\textbf{The Geometry of the Underlying Manifold:} Characteristic classes are calculated over manifolds. In this case, that manifold is the $2$-torus whose symmetries can be incorporated into a CNN with only minor modification. On the other hand, there may be cases where characteristic classes need to be calculated over more complicated spaces. One can imagine such a setting being an ideal application for recently developed methods from geometric deep learning that can capture the structure of non-Euclidean spaces.

\textbf{$U(1)$-invariance:} The data stored over any point in a vector bundle is a vector space. An $n$-dimensional $\mathbf{F}$-vector space is invariant to actions of the group $GL_n(\mathbf{F})$. That is, if one applies $g \in GL_n(\mathbf{F})$ to $V$, one gets other elements of $V$. This symmetry becomes especially important in applications, like data science, where a vector space $V$ needs to be represented concretely by, for example, a basis $b_1, \dots, b_n$. In the line bundle case of the Haldane Bundles dataset, the line $l$ is represented by a non-zero vector $v$ that $l$ passes through it. It is understood by the mathematician that for any non-zero $a \in \mathbf{F}$, $cv$ represents the same line, but this is likely not true for the model. It may be interesting to think about whether specialized augmentation or equivariant architectures might help build a more robust model.

\textbf{Global properties from local characteristics:} Topological characteristics are intrinsically global descriptors of a space. That is, they can rarely be consistently inferred if one only examines part of a space. While at its best, deep learning is powerful precisely because it can extract high-level features by aggregating low-level ones, it is also known to catastrophically fail by focusing on spurious correlations. Developing methods of nudging a model away from single feature correlations is critical in the problem of predicting characteristic classes and could be usefully applied in a host of other problems.

\section{Conclusion}

In this paper we introduced the dataset Haldane Bundles for the purpose of developing better machine learning techniques of predicting the characteristic classes of a vector bundle. While our ultimate motivation comes from the importance of characteristic classes in a range of applications (most importantly, topological materials), we also aim to show that this is an interesting problem from the perspective of incorporating sophisticated mathematical ideas into deep learning.

\begin{ack}
This research was supported by the Mathematics for Artificial Reasoning in Science (MARS) initiative at Pacific Northwest National Laboratory.
It was conducted under the Laboratory Directed Research and Development (LDRD) Program at at Pacific Northwest National Laboratory (PNNL), a multiprogram
National Laboratory operated by Battelle Memorial Institute for the U.S. Department of Energy under Contract
DE-AC05-76RL01830.
\end{ack}



\bibliographystyle{plain}
\bibliography{neurips}

\appendix


\section{Further Constructions of Haldane Bundles}

Whenever one obtains a class of line bundles, one can start taking tensor products to obtain even more varied line bundles.  Furthermore, the Chern number of a tensor product of line bundles is just the sum of the Chern numbers of each piece.  This gives us a way to calculate the Chern number of the tensor product of Haldane bundles, for which we already have a way to calculate.

First, we will generalize a Haldane pair, to a sequence to represent the individual tensor products.

\begin{definition}
A Haldane sequence is a sequence of Haldane pairs $Y=((G_1,F_1),\dots, (G_n,F_n))$, and the corresponding Line bundle is
\begin{align*}
L(Y) = L(G_1,F_1)\otimes\cdots\otimes L(G_n,F_n).
\end{align*}
\end{definition}

By theorem \ref{thm: Chern number} and the fact that Chern number is additive over tensor products, we can calculate the Chern number for any Haldane sequence $Y=((G_1,F_1),\dots, (G_n,F_n))$ as
\begin{align*}
c_1^{\#}(Y) = \sum_{j=1}^n c_1^{\#}(L(G_j,F_j)).
\end{align*}
These type of line bundles we do not use in our data sets, but we do include them in our code.  They give a way to get larger Chern number by using smaller degrees. But, to represent the data for models to learn on they can be more complicated as they have more components since these line bundles are represented by the smooth maps $\Psi:M\rightarrow \CCP^{2^n-1}$ for $n\geq 1$.   

\section{Proofs}
\subsection{Proof of Theorem \ref{thm: Chern number}}
\label{appendix:proof_thm1}

For this section, we will give a short construction of the curvature form used in computing the Chern number and give a proof to theorem \ref{thm: Chern number} .
    
\begin{definition}
A \emph{pre-line bundle} of degree $n$ is a triple $X=(V_{\alpha},\psi_{\alpha}, \tau_{\alpha,\beta})_{\alpha,\beta\in A}$ on some arbitrary index set $A$ consisting of the following properties:
\begin{itemize}
    \item the maps $\psi_{\alpha}:V_{\alpha}\rightarrow \CC^{n+1}\setminus\{0\}$ and $\tau_{\alpha,\beta}:V_{\alpha}\cap V_{\beta}\rightarrow \CC^*$ are smooth for all $\alpha,\beta\in A$,
    \item for each $\alpha,\beta\in A$, $\psi_{\alpha}(p)=\tau_{\alpha,\beta}(p)\psi_{\beta}(p)$ for all $p\in V_{\alpha}\cap V_{\beta}$ and all $\alpha,\beta\in A$,
    \item and $(V_{\alpha})_{\alpha\in A}$ covers $M$. 
\end{itemize}
\end{definition}

As in section \ref{sect: construction}, if $X=(V_{\alpha},\psi_{\alpha},\tau_{\alpha,\beta})$ is a pre-line bundle, we can build a smooth map $\Psi_X:M\rightarrow \CCP^n$ with $\Psi(p) = [\psi_{\alpha}(p)]$ for $p\in V_{\alpha}$.  Let $L_n$ be the tautological line bundle on $\CCP^n$ with
\begin{align*}
    L_n = \{(p,v)~:~ p\in \CCP^n,~~v\in l_p\}
\end{align*}
where $l_p$ is the $1$-dimensional vector space spanned by a representative of $p$.  Then the corresponding line bundle is defined as $L(X):=\Phi_X^*(L_n)$.

To build our curvature, we first need to build a connection on $L(X)$, and we need to make sure our maps $\psi_{\alpha},\tau_{\alpha,\beta}$ are normalized. We say that a pre-line bundle $X=(V_{\alpha},\psi_{\alpha},\tau_{\alpha,\beta})$ of degree $n$ is normalized if $|\psi_{\alpha}(p)|=1$ for all $p\in V_{\alpha}$ and $|\tau_{\alpha,\beta}(p)|=1$ in the complex plane for all $p\in V_{\alpha}\cap V_{\beta}$ and for all $\alpha,\beta\in A$.

\begin{lemma}
Let $X=(V_{\alpha},\Psi_{\alpha},\tau_{\alpha,\beta})_{\alpha\in A}$ be a  complex line bundle of degree $n$.  For each $V_{\alpha}$, there is a local frame $\Sigma_{\alpha}:V_{\alpha}\rightarrow L(X)$ such that on $V_{\alpha}\cap V_{\beta}$ we have the following transformation law
\begin{align*}
    \Sigma_{\alpha} = \tau_{\alpha,\beta}\Sigma_{\beta}
\end{align*}
for all $\alpha,\beta\in A$.
\end{lemma}

\begin{proof}
Let $\Psi_X:M\rightarrow \CCP^n$ be the induced smooth map from the pre-line bundle $X$. For each $\alpha\in A$, define
\begin{align*}
    \Sigma_{\alpha}(p) = (p,\Psi(p),(\psi_{\alpha}^0(p),\dots,\psi_{\alpha}^n(p) ))
\end{align*}
for $p\in V_{\alpha}$, which is a smooth frame on $V_{\alpha}$.  It is easy to see that we have
\begin{align*}
    \tau_{\alpha,\beta}\Sigma_{\beta} &= (p,\Psi(p), \tau_{\alpha,\beta}(\psi_{\beta}^0(p),\dots, \psi_{\beta}^n(p))\\
    &= (p,\Psi(p), (\psi_{\alpha}^0(p),\dots, \psi_{\alpha}^n(p)))\\
    &=\Sigma_{\alpha}.
\end{align*}
This completes the proof.
\end{proof}

Next, we will define our connection form on each of the local sections $V_{\alpha}$ and show that they transform in the expected way for connections.  For this definition, we define $\Omega_M^n(U,\CC)$ to be the space of complex differential $n$-forms on an open subset $U$ of $M$.

\begin{definition}\label{def: connection}
Let $X=(V_{\alpha},\psi_{\alpha},\tau_{\alpha,\beta})_{\alpha,\beta\in A}$ be a normalized complex line bundle of degree $n$ on a smooth manifold $M$.  For each $V_{\alpha}$, define the $1$-form $\omega_{\alpha}\in \Omega_M^1(V_{\alpha},\CC)$ as
\begin{align*}
    \omega_{\alpha} =\sum_{r=1}^n\overline{\psi_{\alpha}^r} d\psi_{\alpha}^r
\end{align*}
for all $\alpha\in A$.  
\end{definition}

\begin{lemma}
Let $X=(V_{\alpha},\Psi_{\alpha},\tau_{\alpha,\beta})_{\alpha,\beta\in A}$ be a normalized complex line bundle of degree $n$ over a  smooth manifold $M$.  The collection of one forms $\omega^B_{\alpha}$ on each local trivialization $V_{\alpha}$ defines a global connection $\nabla:\Gamma(L(X)\rightarrow \Gamma((T^*M)_{\CC}\otimes_{\CC} L(X))$ with 
\begin{align*}
    \nabla(\Sigma_{\alpha}) = \omega_{\alpha}\otimes \Sigma_{\alpha}
\end{align*}
on each open subset $V_{\alpha}$.
\end{lemma}

\begin{proof}
We need to show that the collection of $1$-forms $\omega_{\alpha}$ transform like a connection.  On $V_{\alpha}\cap V_{\beta}$, we have
\begin{align*}
    \omega_{\alpha} &= \sum_{r=1}^n \overline{\psi_{\alpha}^r} d\psi_{\alpha}^r\\
    &= \sum_{r=1}^n \overline{\tau}_{\alpha,\beta} \overline{\psi_{\beta}^r} d(\tau_{\alpha,\beta}\psi_{\beta}^r)\\
    &=\sum_{r=1}^n \overline{\tau}_{\alpha,\beta} \overline{\psi_{\beta}^r}\tau_{\alpha,\beta} d\psi_{\beta}^r + \sum_{r=1}^n \overline{\tau}_{\alpha,\beta} \overline{\psi_{\beta}^r}\psi_{\alpha}^rd\tau_{\alpha,\beta}\\
    &=\sum_{r=1}^n \overline{\psi_{\beta}^r} d\psi_{\beta}^r + \sum_{r=1}^n (\overline{\psi_{\beta}^r}\psi_{\beta}^r ) \overline{\tau}_{\alpha,\beta} d\tau_{\alpha,\beta}\\
    &=\omega_{\beta} + \overline{\tau}_{\alpha,\beta}d\tau_{\alpha,\beta}.
\end{align*}
Multiplying both sides by $\tau_{\alpha,\beta}$ will get the result.  
\end{proof}

For each $\alpha$, we can compute $\Omega_{\alpha} = d\omega_{\alpha}\in\Omega^2_M(V_{\alpha},\CC)$, and these glue together to make a global curvature form $\Omega$ on $M$.  This shows that the curvature is of the form
\begin{align*}
    \Omega|_{V_{\alpha}} = \sum_{r=0}^n d\overline{\psi_{\alpha}^r}\wedge d\psi_{\alpha}^r
\end{align*}
for each $\alpha\in A$.  By Chern-Weil theory the curvature gives us a Chern class $c_1(X)=[\frac{i\Omega}{2\pi}]\in H_{dR}^2(M;\CC)$ and hence gives us the chern number
\begin{align*}
    c_1^{\#}(X) = \frac{i}{2\pi}\int_M\Omega 
\end{align*}
whenever $M$ is a $2$-dimensional manifold.

Now that we have our curvature form, we can use this to give the proof of theorem \ref{thm: Chern number}.

\begin{proof}[proof of theorem \ref{thm: Chern number}]
Let $(G,F)$ be a Haldane pair with our smooth maps $\psi:V\rightarrow \CC^2\setminus\{0\}$ and $\psi^{\dagger}:V^{\dagger}\rightarrow \CC^2\setminus\{0\}$ defined as in \ref{sect: construction}.  We need to normalize these vectors to use our connection in definition \ref{def: connection}, and it is not too hard to show that $|\psi| = \sqrt{2R(R-G)}$ and $|\psi^{\dagger}| = \sqrt{2R^{\dagger}(R^{\dagger}-G)}$. We define
\begin{align*}
    &\widetilde{\psi} = \delta\psi & \widetilde{\psi^{\dagger}} = \delta^{\dagger}\psi^{\dagger}
\end{align*}
with $\delta = \frac{1}{\sqrt{2R(R-G)}}$ and $\delta^{\dagger} = \frac{1}{\sqrt{2R^{\dagger}(R^{\dagger}-G})}.$. 

Now that we normalized our vectors, we can use these to compute our curvature.  On $V$, we have
\begin{align*}
    \Omega|_{V} &= d(\overline{\delta R})\wedge d(\delta R) + d(\overline{\delta F}) \wedge d(\delta F)\\
    &=d(\delta \overline{F})\wedge d(\delta F)\\
    &=\delta^2 d\overline{F}\wedge F + \delta F d\overline{F}\wedge d\delta + \delta \overline{F} d\delta \wedge dF
\end{align*}
using the fact that $R$ and $\delta$ are real and the product rule for the exterior derivative.  Similarly, we obtain
\begin{align*}
    \Omega|_{V^{\dagger}} &= d(\delta^{\dagger} F)\wedge d(\delta^{\dagger}\overline{F})\\
    &=(\delta^{\dagger})^2 dF \wedge d\overline{F} + \delta^{\dagger}Fd\delta^{\dagger}\wedge d\overline{F} + \delta^{\dagger}\overline{F} dF \wedge d\delta^{\dagger}.
\end{align*}
We can break our integral apart on the two open subset $V$ and $V^{\dagger}\setminus \overline{V\cap V^{\dagger}}$ to get the result.
\end{proof}
    
\subsection{Computation of Chern number over the Torus} 
\label{appendix:computation}
In the case where $(G,F)$ is a Haldane pair over $\Torus^2$ consisting of Fourier polynomials, we can approximate the Chern number using the formula in theorem \ref{thm: Chern number} in a very efficient manner on the GPU.  First, we need to give a description on how we save our data of Fourier polynomials in the computer, as this will be important for computation later.  Let 
\begin{align}
    F(p) = \frac{1}{2\pi}\sum_{k\in\ZZ^2} c_ke^{2\pi i(k\cdot p)}
\end{align}
for $c_k\in\CC$ where all but a finite number of them are nonzero, and let
\begin{align}
    G = \frac{1}{2\pi}\sum_{k\in\ZZ^2} \left(a_k\cos(2\pi (k\cdot p)) + b_k \sin(2\pi (k\cdot p))\right)
\end{align}
for $a_k,b_k\in\RR$, where all but a finite number of them are zero. 

Next, we need to describe some bounds to when the coefficients are zero.  Let $N_-(F)$ and $N_+(F)$ be positive integers such that $N_-(F)$ is the minimal integer such that $c_k=0$ for all $k_1<-N_-(F)$ and $k_2<-N_-(F)$, and define $N_+(F)$ to be the minimal positive integer such that $c_k=0$ for all $k_1>N_+(F)$ and $k_2>N_+(F)$.  Similarly, we can define $N_-(G)$ and $N_+(G)$ as well applied to both coefficients $a_k$ and $b_k$.  
Let $N_F=N_-(F) + N_+(F)+1$ and construct a matrix $A_F$ of size $N_F\times N_F$ with entries $(A_F)_{i,j} = c_{(i-N_-(F),j-N_-(F))}$ for $1\leq i,j\leq N_F$.  Similarly, define $N_G = N_-(G) + N_+(G)+1$ and define the matrices $B_G^t$ of size $N_G\times N_G$ for $t=1,2$ with
\begin{align}
    &(B_G^1)_{i,j} = a_{(i-N_-(G),j -N_-(G)}&
    (B_G^2)_{i,j} = b_{(i-N_-(G),j-N_-(G)}
\end{align}
for $1\leq i,j\leq N_G$.  For example, if $N_-(F)=N_+(F)=d$, then $A_F$ is of the form
\begin{align}
    \begin{bmatrix}
    c_{(-d,-d)}&\cdots & c_{(0,-d)}&\cdots & c_{(d,-d)}\\
    \vdots&\cdots & \vdots&\cdots & \vdots\\
    c_{(-d,0)}&\cdots & c_{(0,0)}&\cdots & c_{(d,0)}\\
    \vdots&\cdots & \vdots&\cdots & \vdots\\
    c_{(-d,d)}&\cdots & c_{(0,d)}&\cdots & c_{(d,d)}
    \end{bmatrix}.
\end{align}

\begin{example}
Here is a simple example, that is related to the Haldane model.  Suppose we have a Haldane pair $(G,F)$ with the Fourier polynomials
\begin{align*}
    F(p)  &= t_1e^{2\pi i (x+y)} + t_1e^{2\pi i (-y)} + t_1e^{2\pi i (-x)}\\
    G(p) &= M + 2t_2 \sin(-2\pi x + 2\pi y) + 2t_2 \sin(4\pi x + 2\pi y) + 2t_2 \sin(-2\pi x - 4\pi y)\\
\end{align*}
with parameters $t_1,t_2,M\in\RR$.  Then
then $N_-(F) = N_+(F) = 1$ and $N_-(G)=N(G,-)=2$ and we would have the following matrices for the coefficients:
\begin{align*}
    &A_F = \begin{bmatrix}
    0 & t_1 & 0\\
    t_1 & 0 & 0\\
    0 & 0 & t_1
    \end{bmatrix}
    & B_G^1 = \begin{bmatrix}
    0 & 0 & 0 & 0 & 0\\
    0 & 0 & 0 & 0 & 0\\
    0 & 0 & M & 0 & 0\\
    0 & 0 & 0 & 0 & 0\\
    0 & 0 & 0 & 0 & 0
    \end{bmatrix},~
    & B_G^2 = \begin{bmatrix}
    0 & 2t_2 & 0 & 0 & 0\\
    0 & 0 & 0 & 0 & 0\\
    0 & 0 & 0 & 0 & 0\\
    0 & 2t_2 & 0 & 0 & 2t_2\\
    0 & 0 & 0 & 0 & 0
    \end{bmatrix}.
\end{align*}
\end{example}

With our representation of our Fourier polynomials through square matrices containing their coefficients, we can start describing how we calculated the Chern number using theorem \ref{thm: Chern number}. We proceed in the usual method by partitioning our torus in to small enough squares and evaluate our $\Omega$ on each of the sampled points of each square, multiplying by the area of that square, and then taking the sum to approximate the integral.  The most expensive part of all of this is finding the values of $F$,$G$, and the partial derivatives at each of the sampled points and adding up the values.  As the partition size of the torus grows, the more computationally expensive this becomes.  To handle this, it is fruitful to calculate the values of $F$,$G$ and the partial derivatives using the parallel capabilities of GPU's.

First, partition our torus into $1/N\times 1/N$ sized squares for some large enough $N>1$, so that there are $N^2$ boxes in total, and we have sampled a point in each of these boxes, say the bottom left corner of each of them.
 Let $P=[0,\dots, i/N,\dots, (N-1)/N]$ be a vector with $N$ entries essentially representing either the first or second entry of a sampled point in the torus. Furthermore, let $K_F = [-N_-(F),\dots, 0,\dots, N_+(F)]$ and $K_G = [-N_-(G),\dots, 0,\dots, N_+(G)]$ be two vectors representing the possible entries for $k\in\ZZ^2$ where $a_k,b_k,c_k$ are possibly non-zero. We will be using these vectors with certain operations to make it possible to implement the computation on the GPU. 

Before we start the computation, we will review the Kronecker product, Kronecker sum of two vectors, and a new operation that combines these together which makes it possible to evaluate the sampled points on the Torus using the GPU.
\begin{definition}
Let $v=[v_1,\dots, v_n],h=[h_1,\dots, h_m]$ two vectors with components in an arbitrary field $\FF$.  The Kronecker product of $v$ and $h$ is the vector $w=\text{Kron}(v\otimes h)$ of length $nm$ with
\begin{align}
    w &= [v_1h_1, v_1h_2,\dots, v_1h_m,\dots, v_nh_1,\dots, v_nh_m]\\
    &= [v_1h,\dots, v_nh].
\end{align}
\end{definition}
The Kronecker product is essentially a ordered list of all possible multiplications of elements from $v$ and $h$.  Note that $Kron:\FF^n\otimes \FF^m\rightarrow \FF^{nm}$ is a linear map over $\FF$ and it is not commutative.

Next, is the Kronecker sum, which is similar to the Kronecker product, except it calculates all possible sums between two vectors and arranges it into a matrix.
\begin{definition}
Let $v=[v_1,\dots, v_n]$ and $h=[h_1,\dots, h_m]$ be two vectors with components in an arbitrary field $\FF$.  The Kronecker sum of $v$ and $h$ is the matrix $C=\text{Kron}_{\Sigma}(v,h)$ defined as follows. Let $A$ be a matrix of size $m\times n$ with each row is the vector $v$ as in
\begin{align}
A = \begin{bmatrix}
v\\
\vdots\\
v
\end{bmatrix}.
\end{align}
Similarly, let $B$ be a matrix of size $m\times n$ with each column is the vector $h^T$ as in
\begin{align}
B = \begin{bmatrix} h^T & \cdots & h^T\end{bmatrix}.
\end{align}
We define $C$ to be the $m\times n$ matrix
\begin{align}
 C = A + B & = \begin{bmatrix}
    v_1 + h_1  & v_2 + h_1 & \cdots & v_n + h_1\\
    v_1 + h_2 & v_2 + h_2 & \cdots & v_n + h_2\\
    \vdots & \vdots & \vdots & \vdots\\
    v_1 + h_m & v_2 + h_m & \cdots & v_n + h_m
    \end{bmatrix}.
\end{align}
\end{definition}
The Kronecker sum $Kron_{\Sigma}:\FF^n\times \FF^m\rightarrow Mat_{m,n}(\FF)$ is linear over $\FF$ and it is also non-commutative.

With these two operations, we define a new linear map $\text{mdot}:\FF^n\otimes\FF^n\rightarrow Mat_{nm}(\FF)$ defined for $v\in\FF^n$ and $h\in\FF^m$ as
\begin{align}
    \text{mdot}(v\otimes h) &= Kron_{\Sigma}(Kron(v,h), Kron(v,h))\\
    &=\begin{bmatrix}
    v_1h_1 + v_1h_1  & \cdots & v_1h_m + v_1h_1 & \cdots & v_nh_1 + v_1h_1  & \cdots & v_nh_m + v_1h_1\\
    v_1h_1 + v_1h_2  & \cdots & v_1h_m + v_1h_2 & \cdots & v_nh_1 + v_1h_2  & \cdots & v_nh_m + v_1h_2\\
    \vdots & \cdots & \vdots & \cdots & \vdots  & \cdots & \vdots\\
    v_1h_1 + v_1h_m & \cdots & v_1h_m + v_1h_m & \cdots & v_nh_1 + v_1h_m  & \cdots & v_nh_m + v_1h_m\\
    \vdots & \cdots & \vdots & \cdots & \vdots  & \cdots & \vdots\\
    v_1h_1 + v_nh_1 & \cdots & v_1h_m + v_nh_1 & \cdots & v_nh_1 + v_nh_1  & \cdots & v_nh_m + v_nh_1\\
    \vdots & \cdots & \vdots & \cdots & \vdots  & \cdots & \vdots\\
    v_1h_1 + v_nh_m &  \cdots & v_1h_m + v_nh_m & \cdots & v_nh_1 + v_nh_m  & \cdots & v_nh_m + v_nh_m\\
    \end{bmatrix}
\end{align}
which essentially a dot product between all possible combinations $(v_i,v_j)\cdot (h_r,h_t)$ for $1\leq i,j\leq n$ and $1\leq r,t\leq m$.  One very useful property of this product is that we can break this matrix into a $n\times n$ block matrix of the form
\begin{align}
    mdot(v,h) = \begin{bmatrix} A_{1,1} & A_{2,1} & \cdots & A_{n,1}\\
    A_{1,2} & A_{2,2} & \cdots & A_{n,2}\\
    \vdots & \vdots & \cdots & \vdots\\
    A_{1,n} & A_{2,n} & \cdots & A_{n,n}
    \end{bmatrix}
\end{align}
where
\begin{align}
    A_{i,j} = \begin{bmatrix}
    (v_i,v_j) \cdot (h_1,h_1) & (v_i,v_j)\cdot (h_2,h_1) & \cdots & (v_i,v_j)\cdot (h_m,h_1)\\
    (v_i,v_j) \cdot (h_1,h_2) & (v_i,v_j)\cdot (h_2,h_2) & \cdots & (v_i,v_j)\cdot (h_m,h_2)\\
    \vdots & \vdots & \cdots & \vdots\\
    (v_i,v_j) \cdot (h_m,h_1) & (v_i,v_j)\cdot (h_m,h_2) & \cdots & (v_i,v_j) \cdot (h_m,h_m)
    \end{bmatrix}.
\end{align}
  Note that this operation is non-commutative, and the ordering we do the product in the next section is very important.  Furthermore, since all of these operations are operations that can be done in parallel in the computer, hence $\text{mdot}$ is able to be computed in the GPU.

With our new operation $\text{mdot}$ at our disposal, we can continue our computation of our Chern number on the torus. Our goal was to compute the values of $F,G$,and the partial derivatives on the whole torus.  The first essential part of this, is to compute the matrices $\text{mdot}(P\otimes K_F)$ for $F$, $\text{mdot}(P\otimes K_G)$ for $G$, and similar products for the partial derivatives.  We will focus our attention on the values of $F$ on the torus, as it is similar for $G$ and the partial derivatives.
    
By definition $\text{mdot}(P\otimes K_F)$ is a $N\times N$ block matrix $(A_{i,j})_{i,j}$ described in the definition of $\text{mdot}$ above.  Apply the exponential function component wise to get matrix 
\begin{align}
    \exp{(2\pi i \text{mdot}(P\otimes K_F))} = \begin{bmatrix} \exp{2\pi iA_{1,1}} & \exp{2\pi iA_{2,1}} & \cdots & \exp{2\pi iA_{n,1}}\\
    \exp{2\pi iA_{1,2}} & \exp{2\pi iA_{2,2}} & \cdots & \exp{2\pi iA_{n,2}}\\
    \vdots & \vdots & \cdots & \vdots\\
    \exp{2\pi iA_{1,n}} & \exp{2\pi iA_{2,n}} & \cdots & \exp{2\pi iA_{n,n}}
    \end{bmatrix}.
\end{align}
Here, the matrix $\exp{2\pi i A_{i,j}}$ represents all possible values $e^{2\pi i (k\cdot p)}$ for a fixed point $p=(\frac{i-1}{N},\frac{j-1}{N})$ and all possible $k\in\ZZ^2$ with respect to $N_+(F)$ and $N_-(F)$.
Applying a discrete convolution with weight $A_F$, input $\exp{(2\pi i \text{mdot}(P,K_F))}$, and stride $N$ results in a $N\times N$ matrix
\begin{align}
    \begin{bmatrix}
    F(0,0) & F(\frac{1}{N},0) & \cdots & F(\frac{N-1}{N}, 0)\\
    F(0,\frac{1}{N}) & F(\frac{1}{N},\frac{1}{N}) & \cdots & F(\frac{N-1}{N},\frac{1}{N})\\
    \vdots & \vdots & \cdots & \vdots\\
    F(0,\frac{N-1}{N}) & F(\frac{1}{N},\frac{N-1}{N}) & \cdots & F(\frac{N-1}{N},\frac{N-1}{N})
    \end{bmatrix}
\end{align}
which are the values of $F$ on each of the sampled points of our partition.

With this and applying the same idea for $G$ and the partial derivatives, we get our values of $F$,$G$, and the partial derivatives on the sampled points of the torus.  We can use this to compute the integrands in the GPU and take the sum at the end to obtain an approximation to the integral.

\end{document}